\documentclass[english,runningheads]{ASMDA}
\usepackage[T1]{fontenc}
\usepackage[latin9]{inputenc}
\usepackage{amsthm}
\usepackage{amsmath}
\usepackage{amssymb}
\usepackage{esint}

\makeatletter
\theoremstyle{plain}
\newtheorem{thm}{\protect\theoremname}
\theoremstyle{plain}
\newtheorem{prop}[thm]{\protect\propositionname}
\ifx\proof\undefined
\newenvironment{proof}[1][\protect\proofname]{\par
\normalfont\topsep6\p@\@plus6\p@\relax
\trivlist
\itemindent\parindent
\item[\hskip\labelsep
\scshape
#1]\ignorespaces
}{%
\endtrivlist\@endpefalse
}
\providecommand{\proofname}{Proof}
\fi
\theoremstyle{plain}
\newtheorem{cor}[thm]{\protect\corollaryname}

\title*{On some interrelations of generalized $q$-entropies and a generalized Fisher information, including a Cramér-Rao inequality
\thanks{To be presented during the conference ``Applied Stochastic Models and Data Analysis'', Matar\'o (Barcelona), Spain, 25-28 June 2013 (and a version published in the proceedings)}}
\toctitle{Interrelations of generalized q-entropies and a generalized Fisher information}
\titlerunning{Interrelations of generalized q-entropies and a generalized Fisher information}
\author{
Jean-François Bercher\inst{1}
}
\index{Bercher, J.-F.}

\authorrunning{Bercher}
\institute{
Laboratoire d'informatique Gaspard Monge, UMR 8049 ESIEE-Paris, Université Paris-Est, France \\
(E-mail: {\tt jf.bercher@esiee.fr})
}

\makeatother

\usepackage{babel}
\providecommand{\corollaryname}{Corollary}
\providecommand{\propositionname}{Proposition}
\providecommand{\theoremname}{Theorem}

\begin{document}
\maketitle
\begin{abstract}
In this communication, we describe some interrelations between generalized
$q$-entropies and a generalized version of Fisher information. In
information theory, the de Bruijn identity links the Fisher information
and the derivative of the entropy. We show that this identity can
be extended to generalized versions of entropy and Fisher information.
More precisely, a generalized Fisher information naturally pops up
in the expression of the derivative of the Tsallis entropy. This generalized
Fisher information also appears as a special case of a generalized
Fisher information for estimation problems. Indeed, we derive here
a new Cramér-Rao inequality for the estimation of a parameter, which
involves a generalized form of Fisher information. This generalized
Fisher information reduces to the standard Fisher information as a
particular case. In the case of a translation parameter, the general
Cramér-Rao inequality leads to an inequality for distributions which
is saturated  by generalized $q$-Gaussian distributions. These generalized
$q$-Gaussians are important in several areas of physics and mathematics.
They are known to maximize the $q$-entropies subject to a moment
constraint. The Cramér-Rao inequality shows that the generalized $q$-Gaussians
also minimize the generalized Fisher information among distributions
with a fixed moment. Similarly, the generalized $q$-Gaussians also
minimize the generalized Fisher information among distributions with
a given $q$-entropy.  \keyword{Cramér-Rao inequality, generalized $q$-entropy, generalized Gaussians, de Bruijn identity}
\end{abstract}
Let $f(x)$ be a probability distribution defined on $X\mathbb{\subseteq R}^{n}$.
If $M_{q}[f]=\int_{X}f(x)^{q}\text{d}x$, $q\geq0$ is the information
generating function, then $S_{q}[f]=\frac{1}{1-q}\left(M_{q}[f]-1\right)$
is the so-called Tsallis entropy, or $q$-entropy (which can be easily
related to the Rényi entropy). Both entropies reduce to the standard
Shannon-Boltzmann entropy for $q=1$. It is well-known that the maximum
of the Rényi-Tsallis entropy, among all distributions with a fixed
moment $m_{\alpha}=E_{g}\left[\left\Vert X\right\Vert ^{\alpha}\right],$
is obtained for a generalized $q$-Gaussian distribution 
\begin{equation}
G_{\gamma}(x)\propto\left(1-\left(q-1\right)\gamma\|x\|^{\alpha}\right)_{+}^{\frac{1}{q-1}}\text{ for }q\not=1\label{eq:GenqGauss}
\end{equation}
which reduces to a generalized Gaussian for $q=1.$ It is also well-known
that the classical Fisher information, which is derived from considerations
in estimation theory, is linked to the Shannon entropy via the de
Bruijn identity, and that this Fisher information is minimized by
a Gaussian distribution, among all distributions with a fixed moment
or entropy. It is the main objective of this paper to show that such
results can be extended to the $q$-entropy and a suitable extension
of Fisher information. Incidentally , we derive an interesting new
generalized Cramér-Rao inequality in estimation theory, which might
prove useful in its own right.

\section{An extended de Bruijn identity}

A fundamental connection between the Boltzmann-Shannon entropy, Fisher
information, and the Gaussian distribution is given by the de Bruijn
identity \cite{stam_inequalities_1959}. We show here that this important
connection can be extended to the $q$-entropies, a suitable generalized
Fisher information and the generalized $q$-Gaussian distributions. 

 The de Bruijn identity states that if $Y_{t}=X+\sqrt{2t}Z$ where
$Z$ is a standard Gaussian vector and $X$ a random vector of $\mathbb{R}^{n},$
independent of $Z,$ then 
\begin{equation}
\frac{\text{d}}{\text{d}t}H[f_{Y_{t}}]=I_{2,1}[f_{Y_{t}}]=\phi_{2,1}[f_{Y_{t}}],\label{eq:deBruijnIdentity}
\end{equation}
where $f_{Y_{t}}$ denotes the density of $Y_{t}=X+\sqrt{2t}Z$, and
$I_{2,1}[f_{Y_{t}}]$, $\phi_{2,1}[f_{Y_{t}}]$ are two notations
for the classical Fisher information (the meaning of which will be
made clear in the following). Although the de Bruijn identity holds
in a wider context, the classical proof of the de Bruijn identity
uses the fact that if $Z$ is a standard Gaussian vector, then $Y_{t}$
satisfies the well-known heat equation $\frac{\partial f}{\partial t}=\Delta f,$
where $\Delta$ denotes the Laplace operator.

Nonlinear versions of the heat equation are of interest in a large
number of physical situations, including fluid mechanics, nonlinear
heat transfer or diffusion. Other applications have been reported
in mathematical biology, lubrification, boundary layer theory, etc;
see the series of applications presented in \cite[chapters 2 and 21]{vazquez_porous_2006}
and references therein. The porous medium equation and the fast diffusion
equation  correspond to the differential equation $\frac{\partial f}{\partial t}=\Delta f^{m},$
with $m>1$ for the porous medium equation and $<1$ for the fast
diffusion. These two equations have been exhaustively studied and
characterized by J. L. Vazquez, e.g. in \cite{vazquez_porous_2006,vazquez_smoothing_2006}. 

These equations are included as particular cases into the \textit{doubly
nonlinear equation}, which involves a $p$-Laplacian operator $\Delta_{p}f\,:=\text{div}\left(|\nabla f|^{p-2}\,\nabla f\right),$
and the power $m$ of the porous medium or fast diffusion equation.
This doubly nonlinear equation takes the form
\begin{equation}
\frac{\partial}{\partial t}f=\Delta_{\beta}f^{m}=\text{div}\left(|\nabla f^{m}|^{\beta-2}\,\nabla f^{m}\right),\label{eq:dnle-1}
\end{equation}
where we use $p=\beta$ for convenience and coherence with notation
in the paper. The $\beta$-Laplacian typically appears in the minimization
of a Dirichlet energy like $\int|\nabla f|^{\beta}\text{d}x$ which
leads to the Euler-Lagrange equation. It can be shown, see \cite[page 192]{vazquez_smoothing_2006},
that for $m(\beta-1)+(\beta/n)-1>0$, (\ref{eq:dnle-1}) has a unique
self-similar solution, called a Barenblatt profile, whose initial
value is the Dirac mass at the origin. This fundamental solution is
usually given as a function of $m$. Here, if we put $q=m+1-\frac{\alpha}{\beta}$,
the solution can be written as a $q$-Gaussian distribution: 
\begin{equation}
f(x,t)=\frac{1}{t^{\frac{n}{\delta}}}B\left(\frac{x}{t^{\frac{1}{\delta}}}\right),\,\text{ with }B(x)=\begin{cases}
\left(C-k|x|^{\alpha}\right)_{+}^{\frac{1}{q-1}} & \text{ for }q\neq1\\
\frac{1}{\sigma}\exp\left(-\frac{|\beta-1|}{\beta^{\alpha}}|x|^{\alpha}\right) & \text{ for }q=1
\end{cases}\label{eq:FundamentalSolution}
\end{equation}
with $\delta=n(\beta-1)m+\beta-n>0,\,\,\,\,\, k=\frac{m(\beta-1)-1}{\beta}\left(\frac{1}{\delta}\right)^{\frac{1}{\beta-1}}\text{ and }\alpha=\frac{\beta}{\beta-1}.$

As mentioned above, the doubly nonlinear diffusion equation allows
to derive a nice extension of the de Bruijn identity (\ref{eq:deBruijnIdentity}),
and leads to a possible definition of a generalized Fisher information.
This is stated in the next Proposition. The case $\beta=2$ of this
result has been given in a paper by Johnson and Vignat \cite{johnson_results_b_2007}. 
\begin{prop}
{[}Extended de Bruijn identity \cite{bercher:hal-00766699}{]} Let
$f(x,t)$ a probability distributions defined on a subset $X$ of
$\mathbb{R}^{n}$ and satisfying the doubly nonlinear equation (\ref{eq:dnle-1}).
Assume that the domain $X$ is independent of $t,$ that $f(x,t)$
is differentiable with respect to $t,$ is continuously differentiable
over $X$, and that $\frac{\partial}{\partial t}f(x,t)^{q}$ is absolutely
integrable and locally integrable with respect to $t$. Then, for
$\beta>1,$ $\alpha$ and $\beta$ Hölder conjugate of each other,
for $q=m+1-\frac{\alpha}{\beta}$, $M_{q}[f]=\int f^{q}$ and $S_{q}[f]=\frac{1}{1-q}\left(M_{q}[f]-1\right)$
the Tsallis entropy, we have 
\begin{align}
\frac{\text{d}}{\text{d}t}S_{q}[f] & =q\, m^{\beta-1}\phi_{\beta,q}[f]=\left(\frac{m}{q}\right)^{\beta-1}M_{q}[f]^{\beta}\, I_{\beta,q}[f]\label{eq:ExtendedDeBruijnb}
\end{align}
\begin{equation}
\text{with \,\,\,\,\,}\phi_{\beta,q}[f]=\int_{X}f(x)^{\beta(q-1)+1}\left(\frac{|\nabla f(x)|}{f(x)}\right)^{\beta}\mathrm{d}x\text{{\,\,\ and\,\,}}I_{\beta,q}[f]=\frac{\phi_{\beta,q}[f]}{M_{q}[f]^{\beta}}.\label{eq:defsFisher}
\end{equation}

\end{prop}
In (\ref{eq:defsFisher}), $\phi_{\beta,q}[f]$ and $I_{\beta,q}[f]$
are two possible generalization of Fisher information. Of course,
the standard Fisher information is recovered in the particular case
$\alpha=\beta=2,$ and $q=m=1,$ and so is the de Bruijn identity
(\ref{eq:deBruijnIdentity}). The proof of this result relies on integration
by part (actually using the Green identity) along the solutions of
the nonlinear heat equation (\ref{eq:dnle-1}). This proof can be
found in \cite{bercher:hal-00766699} and is not repeated here. A
variant of the result for $\beta=2$, which considers a free-energy
instead of the entropy above, is well-known in certain circles, see
e.g. \cite{dolbeault_improved_????,carrillo_asymptotic_2000}. More
than that, by using carefully the calculations in \cite{dolbeault_improved_????},
it is possible to check that $\frac{\text{d}}{\text{d}t}\phi_{2,q}[f]\leq0$
for $q>1-\frac{1}{n},$ which means the Tsallis entropy is a monotone
increasing concave function along the solutions of (\ref{eq:dnle-1}).
In their recent work \cite{savare_concavity_2012}, Savaré and Toscani
have shown that in the case $\beta=2,$ $m=q$, the entropy power,
up to a certain exponent, is a concave function of $t,$ thus generalizing
the well-known concavity of the (Shannon) entropy power to the case
of $q$-entropies. This allows to obtain as a by-product a generalized
version of the Stam inequality, valid for the solutions of (\ref{eq:dnle-1}).
We will come back to this generalized Stam inequality in Proposition
6.

\section{Extended Cramér-Rao inequalities}

Let $f(x;\theta)$ be a probability distribution, with $x\in X\subseteq\mathbb{R}^{n}$
and $\theta\in\mathbb{R}^{k}$. We will deal here with the estimation
of a scalar function $h(\theta)$ of $\theta,$ with $T(x)$ the corresponding
estimator (the more general case where $h(\theta)$ and $T(x)$ are
vector valued is a bit more involved; some results are given in \cite{bercher:hal-00766695}
with general norms). We extend here the classical Cramér-Rao inequality
in two directions: firstly, we give results for a general moment of
the estimation error instead of the second order moment, and secondly
we introduce the possibility of computing the moment of this error
with respect to a distribution $g(x;\theta)$ instead of $f(x;\theta)$:
in estimation, the error is $T(X)-h(\theta),$ and the bias can be
evaluated as $\int_{X}\left(T(x)-h(\theta)\right)\, f(x;\theta)\,\text{d}x=\mathrm{E}_{f}\left[T(X)-h(\theta)\right]=\eta(\theta)-h(\theta)$,
while a general moment of of the error can be computed with respect
to another probability distribution $g(x;\theta)$, as in $\mathrm{E}_{g}\left[\left|T(X)-h(\theta)\right|^{\beta}\right]=\int_{X}\left|T(x)-h(\theta)\right|^{\beta}\, g(x;\theta)\,\text{d}x.$
The two distributions $f(x;\theta)$ and $g(x,\theta)$ can be chosen
very arbitrary. However, one can also build $g(x;\theta)$ as a transformation
of $f(x;\theta)$ that highlights, or on the contrary scores out,
some characteristics of $f(x;\theta)$.  An important case is when
$g(x;\theta)$ is defined as the escort distribution of order $q$
of $f(x;\theta)$: 

\begin{equation}
f(x;\theta)=\frac{g(x;\theta)^{q}}{\int g(x;\theta)^{q}\text{d}x}\,\,\,\,\text{ and }\,\,\, g(x;\theta)=\frac{f(x;\theta)^{\bar{q}}}{\int f(x;\theta)^{\bar{q}}\text{d}x},\label{eq:PairEscorts-1-1}
\end{equation}
where $q$ is a positive parameter, $\bar{q}=1/q,$ and provided of
course that involved integrals are finite. These escort distributions
are an essential ingredient in the nonextensive thermostatistics context.
It is in the special case where $f(x;\theta)$ and $g(x;\theta)$
are a pair of escort distributions that we will find again the generalized
Fisher information (\ref{eq:defsFisher}) obtained in the extended
de Bruijn identity. Our previous results on generalized Fisher information
can be found in \cite{bercher:hal-00766695,bercher:hal-00733750}
in the case of the direct estimation of the parameter $\theta$. We
propose here a novel derivation, introducing in particular a notion
of generalized Fisher information matrix, in the case of the estimation
of a function of the parameters. Let us first state the result. 
\begin{prop}
Let $f(x;\theta)$ be a multivariate probability density function
defined for $x\in X\mathbb{\subseteq R}^{n}$, and with $\theta\in\Theta\subseteq\mathbb{R}^{k}$
is a parameter of the density.  Let $g(x;\theta)$ denote another
probability density function also defined on $(X;\Theta)$. Assume
that $f(x;\theta)$ is a jointly measurable function of $x$ and $\theta,$
is integrable with respect to $x$, is absolutely continuous with
respect to $\theta,$ and that the derivatives with respect to each
component of $\theta$ are locally integrable. Let $T(x)$ be an estimator
of a function $h(\theta)$ and set $\eta(\theta)=E_{f}[T(X)]$. 
Then, for any estimator $T(x)$ of $h(\theta)$, we have 
\begin{equation}
E_{g}\left[\left|T(X)-h(\theta)\right|^{\alpha}\right]^{\frac{1}{\alpha}}\,\geq\sup_{A>0}\,\frac{\dot{\eta}(\theta)^{T}A\,\dot{\eta}(\theta)}{E_{g}\left[\left|\dot{\eta}(\theta)^{T}A\,\psi_{g}(X;\theta)\right|^{\beta}\right]^{\frac{1}{\beta}}}.\label{eq:CRM_general}
\end{equation}
with equality if and only if $\dot{\eta}(\theta)^{T}A\,\psi_{g}(x;\theta)=c(\theta)\mathrm{sign}(T(x)-h(\theta))\left|T(x)-h(\theta)\right|^{\alpha-1},$
$c(\theta)>0$ and where $\alpha^{-1}+\beta^{-1}=1,$ $\alpha>1$,
$A$ is a positive definite matrix and $\psi_{g}(x;\theta)$ a score
function given with respect to $g(x;\theta):$
\begin{equation}
\psi_{g}(x;\theta):=\frac{\nabla_{\theta}f(x;\theta)}{g(x;\theta)}.\label{eq:def_g_score}
\end{equation}

\end{prop}
 
\begin{proof}
Let $\eta(\theta)=E_{f}[T(X)]$.  Let us first observe that $E_{g}[\psi_{g}(x;\theta)]=\frac{\text{d}}{\text{d}\theta}\int_{X}\, f(x;\theta)\,\text{d}x=0.$
Differentiating $\eta(\theta)=E_{f}[T(X)]$ with respect to each $\theta_{i}$
we get
\begin{alignat*}{1}
\dot{\eta}(\theta)=\nabla_{\theta}\eta(\theta) & =\nabla_{\theta}\int_{X}T(x)\, f(x;\theta)\,\text{d}x\\
 & =\int_{X}T(x)\,\frac{\nabla_{\theta}f(x;\theta)}{g(x;\theta)}\, g(x;\theta)\,\text{d}x\\
 & =\int_{X}\left(T(x)-h(\theta)\right)\,\psi_{g}(x;\theta)\, g(x;\theta)\,\text{d}x.
\end{alignat*}
For any positive definite matrix $A$, multiplying on the left by
$\dot{\eta}(\theta)^{T}A$ gives
\[
\dot{\eta}(\theta)^{T}A\,\dot{\eta}(\theta)=\int_{X}\left(T(x)-h(\theta)\right)\,\dot{\eta}(\theta)^{T}A\,\psi_{g}(x;\theta)\, g(x;\theta)\,\text{d}x,
\]
and by the Hölder inequality, we obtain
\[
E_{g}\left[\left|T(x)-h(\theta)\right|^{\alpha}\right]^{\frac{1}{\alpha}}\, E_{g}\left[\left|\dot{\eta}(\theta)^{T}A\,\psi_{g}(x;\theta)\right|^{\beta}\right]^{\frac{1}{\beta}}\geq\,\dot{\eta}(\theta)^{T}A\,\dot{\eta}(\theta),
\]
with equality if and only if $\left(T(x)-h(\theta)\right)\,\dot{\eta}(\theta)^{T}A\,\psi_{g}(x;\theta)>0$
and $\left|T(x)-h(\theta)\right|^{\alpha}=k(\theta)\left|\dot{\eta}(\theta)^{T}A\,\psi_{g}(x;\theta)\right|^{\beta}$,
$k(\theta)>0$. This inequality, in turn, provides us with the lower
bound (\ref{eq:CRM_general}) for the moment of order $\alpha$ and
computed wrt to $g$ of the estimation error.
\end{proof}
The inverse of the matrix $A$ which maximizes the right hand side
is the Fisher information matrix of order $\beta.$ Unfortunately,
we do not have a closed-form expression for this matrix in the general
case. Two particular cases are of interest. 
\begin{cor}
{[}Scalar extended Cramér-Rao inequality{]} In the scalar case (or
the case of a single component of $\theta$), the following inequality
holds 
\begin{equation}
E_{g}\left[\left|T(X)-h(\theta)\right|^{\alpha}\right]^{\frac{1}{\alpha}}\,\geq\,\frac{\left|\dot{\eta}(\theta)\right|}{E_{g}\left[\left|\psi_{g}(X;\theta)\right|^{\beta}\right]^{\frac{1}{\beta}}},\label{eq:CRscalar_general}
\end{equation}
with equality if and only if \textup{$\psi_{g}(x;\theta)=c(\theta)\mathrm{sign}(T(x)-h(\theta)\,\left|T(x)-h(\theta)\right|^{\alpha-1}$. }
\end{cor}
In the simple scalar case, we see that $A>0$ can be simplified in
(\ref{eq:CRM_general}) and thus that (\ref{eq:CRscalar_general})
follows. Note that for $\alpha=2,$ the equality case implies that
$E_{g}[\psi_{g}]=0=E_{g}[T(X)-h(\theta)],$ which means that $E_{g}[T(X)]=\eta(\theta)=h(\theta),$
i.e. the estimator is unbiased (with respect to both $f$ and $g)$.
Actually, this inequality recovers at once the generalized Cramér-Rao
inequality we presented in the univariate case \cite{bercher:hal-00733750}
. The denominator plays the role of the Fisher information in the
classical case, which corresponds to the case $g(x;\theta)=f(x;\theta)$,
$\beta=2.$ As mentioned above, an extension of this result to the
multidimensional case and arbitrary norms has been presented in \cite{bercher:hal-00766695}
but it does not seem possible to obtain it here as particular case
of (\ref{eq:CRM_general}). 

A second interesting case is the multivariate case $\alpha=\beta=2$.
Indeed, in that case, we get an explicit form for the generalized
Fisher information matrix and an inequality which looks like the classical
one. 
\begin{cor}
{[}Multivariate Cramér-Rao inequality with $\alpha=\beta=2${]} For
$\alpha=\beta=2$, we have

\begin{equation}
E_{g}\left[\left|T(X)-h(\theta)\right|^{2}\right]\geq\dot{\eta}(\theta)^{T}J_{g}(\theta)^{-1}\,\dot{\eta}(\theta)\label{eq:CRM_quadratic}
\end{equation}
with $J_{g}(\theta)=E_{g}\left[\psi_{g}(X;\theta)\psi_{g}(X;\theta)^{T}\right]$,
and with equality if and only if $\left|T(X)-h(\theta)\right|=k(\theta)\left|\dot{\eta}(\theta)^{T}J_{g}(\theta)^{-1}\,\psi_{g}(X;\theta)\right|$.\end{cor}
\begin{proof}
The denominator of (\ref{eq:CRM_general}) is a quadratic form and
we have
\begin{alignat}{1}
E_{g}\left[\left|T(X)-h(\theta)\right|^{2}\right]\, & \geq\sup_{A>0}\,\frac{\left(\dot{\eta}(\theta)^{T}A\,\dot{\eta}(\theta)\right)^{2}}{E_{g}\left[\left|\dot{\eta}(\theta)^{T}A\,\psi_{g}(X;\theta)\right|^{2}\right]}\nonumber \\
 & \geq\sup_{A>0}\,\frac{\left(\dot{\eta}(\theta)^{T}A\,\dot{\eta}(\theta)\right)^{2}}{\dot{\eta}(\theta)^{T}A\, E_{g}\left[\psi_{g}(X;\theta)\psi_{g}(X;\theta)^{T}\right]A^{T}\dot{\eta}(\theta)}.\label{eq:tobesimplified1}
\end{alignat}
Let $J_{g}(\theta)=E_{g}\left[\psi_{g}(X;\theta)\psi_{g}(X;\theta)^{T}\right]$
and set $z(\theta)=A^{\frac{1}{2}}\dot{\eta}(\theta).$ With these
notations, and using the inequality $(z^{T}z)^{2}\leq(z^{T}Bz)\,(z^{T}B^{-1}z)$
valid for any $B>0$, we obtain that 
\[
\dot{\eta}(\theta)^{T}J^{-1}\,\dot{\eta}(\theta)\geq\sup_{A>0}\,\frac{\left(z(\theta)^{T}z(\theta)\right)^{2}}{z(\theta)^{T}A^{\frac{1}{2}}\, J_{g}(\theta)\,\left(A^{\frac{1}{2}}\right)^{T}z(\theta)}.
\]
Since it can be readily checked that the upper bound is attained with
$A=J_{g}(\theta)^{-1},$ we finally end with (\ref{eq:CRM_quadratic}).
Of course, for $g=f,$ the inequality (\ref{eq:CRM_quadratic}) reduces
to the classical multivariate Cramér-Rao inequality. 
\end{proof}
An important consequence of these results is obtained in the case
of a translation parameter, where the generalized Cramér-Rao inequality
induces a new class of inequalities. Let $\theta\in\mathbb{R}$ be
a scalar location parameter, $x\in X\subseteq\mathbb{R}^{n}$, and
define by $f(x;\theta)$ the family of density $f(x;\theta)=f(x-\theta1)$,
where $1$ is a a vector of ones. In this case, we have $\nabla_{\theta}f(x;\theta)=-1^{T}\nabla_{x}f(x-\theta1),$
provided that $f$ is differentiable at $x-\theta1$, and the Fisher
information becomes a characteristic of the information in the distribution.
If $X$ is a bounded subset, we will assume that $f(x)$ vanishes
and is differentiable on the boundary $\partial X$. Without loss
of generality, we will assume that the mean of $f(x)$ is zero. Set
$h(\theta)=\theta$ and take $T(X)=1^{T}X/n,$ with of course $\eta(\theta)=E[T(X)]=\theta$
and $\dot{\eta}(\theta)=1.$ Finally, let us choose the particular
value $\theta=0.$ In these conditions, the generalized Cramér-Rao
inequality (\ref{eq:CRscalar_general}) becomes 
\begin{equation}
E_{g}\left[\left|1^{T}X\right|^{\alpha}\right]^{\frac{1}{\alpha}}\, E_{g}\left[\left|1^{T}\,\frac{\nabla_{x}f(X)}{g(X)}\right|^{\beta}\right]^{\frac{1}{\beta}}\geq n,\label{eq:douze}
\end{equation}
with equality if and only if $1^{T}\,\frac{\nabla_{x}f(x)}{g(x)}=c(\theta)\mathrm{sign}(1^{T}X)\,\left|1^{T}X\right|^{\alpha-1}$.
In\cite{bercher:hal-00766695} , we have a slightly more general result
in the multivariate case: 
\begin{equation}
E_{g}\left[\left\Vert X\right\Vert ^{\alpha}\right]^{\frac{1}{\alpha}}E_{g}\left[\left\Vert \frac{\nabla_{x}f(X)}{g(X)}\right\Vert _{*}^{\beta}\right]^{\frac{1}{\beta}}\geq n,\label{eq:CRInequalityLocation-1}
\end{equation}
where $\|.\|$ is a norm, and the corresponding dual norm is denoted
by $\|.\|_{*}$. Finally, let $f(x)$ and $g(x)$ be a pair of escort
distributions as in (\ref{eq:PairEscorts-1-1}). In such a case, the
following recovers the generalized Fisher information (\ref{eq:defsFisher})
and yields a new characterization of $q$-Gaussian distributions. 
\begin{cor}
{[}$q$-Cramér-Rao inequality{]} Assume that $g(x)$ is a measurable
differentiable function of $x$, which vanishes and is differentiable
on the boundary $\partial X$, and finally that the involved integrals
exist and are finite. Then, for the pair of escort distributions (\ref{eq:PairEscorts-1-1}),
the following $q$-Cramér-Rao inequality holds 
\begin{equation}
q^{\beta}E_{g}\left[\left\Vert X\right\Vert ^{\alpha}\right]^{\frac{1}{\alpha}}\, I_{\beta,q}\left[g\right]^{\frac{1}{\beta}}\geq n,\label{eq:ExtendedqCRloc1}
\end{equation}
\[
\text{with \,\,\,}I_{\beta,q}\left[g\right]=\left(1/M_{q}\left[g\right]^{\beta}\right)\, E\left[g(x)^{\beta(q-1)}\left\Vert \frac{\nabla_{x}g(x)}{g(x)}\right\Vert _{*}^{\beta}\right],
\]
and with equality if and only if $g(x)$ is a generalized $q$-Gaussian,
i.e. $\, g(x)\propto\left(1-\gamma\left\Vert x\right\Vert ^{\alpha}\right)_{+}^{\frac{1}{q-1}},\mathrm{\,\, with}\,\,\gamma>0.$ \end{cor}
\begin{proof}
The result follows from (\ref{eq:CRInequalityLocation-1}), or (\ref{eq:douze})
with $n=1,$ and the fact that 
\[
\frac{\nabla_{x}f(X)}{g(X)}=\frac{q}{M_{q}[g]}\, g(X)^{q-1}\frac{\nabla_{x}g(X)}{g(X)}.
\]
The case of equality is obtained by solving the general equality conditions
in the special case where $f(x)$ and $g(x)$ form a pair of escort
distributions. 
\end{proof}
As a direct consequence of the $q$-Cramér-Rao inequality (\ref{eq:ExtendedqCRloc1}),
we obtain that the minimum of the generalized Fisher information among
all distributions with a given moment of order $\alpha$, say $m_{\alpha}=E_{g}\left[\left\Vert X\right\Vert ^{\alpha}\right],$
is obtained when $g$ is a generalized $q$-Gaussian distribution,
with a parameter $\gamma$ such that the distribution has the prescribed
moment. This parallels, and complements the well known fact that the
$q$-Gaussians maximize the $q$-entropies subject to a moment constraint,
and yields new variational characterizations of the generalized $q$-Gaussians.
As mentioned earlier, the generalized Fisher information also satisfies
an extension of Stam's inequality, which links the generalized Fisher
information and the $q$-entropy power, defined as an exponential
of the Rényi entropy $H_{q}[f]$ as 
\begin{equation}
N_{q}[f]=M_{q}[f]^{\frac{2}{n}\,\frac{1}{1-q}}=\exp\left(\frac{2}{n}H_{q}[f]\right)=\left(\int_{\Omega}f(x)^{q}\text{d}x\right)^{\frac{2}{n}\,\frac{1}{1-q}},\label{eq:defEntropyPower-1}
\end{equation}
for $q\neq1.$ For $q=1,$ we set $N_{q}[f]=\exp\left(\frac{2}{n}H_{1}[f]\right),$
where $H_{1}[f]$ is the Boltzmann-Shannon entropy. The generalized
Stam inequality is given here without proof (see \cite{bercher:hal-00766699}). 
\begin{prop}
{[}Generalized Stam inequality{]} Let $n\geq1,$ $\beta$ and $\alpha$
be Hölder conjugates of each other, $\alpha>1,$ and $q>\max\left\{ (n-1)/n,\, n/(n+\alpha)\right\} $.
Then for any probability density on $\mathbb{R}^{n}$, that is continuously
differentiable, the following generalized Stam inequality holds 
\begin{alignat}{1}
 & I_{\beta,q}\left[f\right]^{\frac{1}{\beta}}\, N_{q}[f]^{\frac{1}{2}}\geq\, I_{\beta,q}\left[G\right]^{\frac{1}{\beta}}\, N_{q}[G]^{\frac{1}{2}}.\label{eq:GeneralizedStamInequalityforI}
\end{alignat}
with $\lambda=n(q-1)+1>0$ and with equality if and only if $f$ is
any generalized $q$-Gaussian (\ref{eq:GenqGauss}).
\end{prop}
The generalized Stam inequality implies that the generalized $q$-Gaussian
minimize the generalized Fisher information within the set of probability
distributions with a fixed $q$-entropy power.  \\\vspace{-0.2cm}

To sum up and emphasize the main results, let us point out that we
have exhibited a generalized Fisher information, both as a by-product
of a generalization of de Bruijn identity and as a fundamental measure
of information in estimation theory. We have shown that this allows
to draw a nice interplay between $q$-entropies, generalized $q$-Gaussians
and the generalized Fisher information. These interrelations yield
the generalized $q$-Gaussians as minimizers of the $(\beta,q)$-Fisher\textquoteright{}s
information under adequate constraints, or as minimizers of functionals
involving $q$-entropies, $(\beta,q)$-Fisher\textquoteright{}s information
and/or moments. This is shown through inequalities and identities
involving these quantities and generalizing classical information
relations (Cramér-Rao\textquoteright{}s inequality, Stam\textquoteright{}s
inequality, De Bruijn\textquoteright{}s identity). 
\end{document}